\documentclass[conference]{IEEEtran}
\usepackage{amsthm, amssymb, amsmath, balance, enumitem, hyperref, makecell, stmaryrd, xfrac, tikz,color}
\newtheorem{definition}{Definition}
\newtheorem{thm}{Theorem}
\newtheorem{rem}{Remark}

\newtheorem{lem}{Lemma}
\newtheorem{prop}{Proposition}

\begin{document}
	\title{Symmetries in Linear Programming for Information Inequalities}
	\author{\IEEEauthorblockN{Emirhan G\"urp\i nar}
		\IEEEauthorblockA{\textit{LIRMM, Universit\'e de Montpellier, CNRS}\\			Montpellier, France \\
			emirhan.gurpinar@lirmm.fr}
		}
	\maketitle
	\begin{abstract}
		We study the properties of secret sharing schemes, where a random secret value is transformed into shares distributed among several participants in such a
		way that only the qualified groups of participants can recover the secret value. We improve the lower bounds on the sizes of shares  for several specific problems of secret sharing. To this end, we use the method of non-Shannon-type information inequalities going back to Z.~Zhang and R.W.~Yeung. We employ and extend the linear programming technique that allows to apply  new information inequalities indirectly,  without even writing them down explicitly. To reduce the  complexity of the problems of linear programming involved in the bounds we extensively use symmetry considerations.
	\end{abstract}
	\section{Introduction}
	Secret sharing was introduced in \cite{B79,Sh79}, see \cite{P13} for a more recent survey. The central problem in this field is to compute for a given access structure its optimal information ratio, i.e.,  the infimum of the size of the largest share in proportion to that of the secret. In general, this problem remains widely open. In this paper we study several particular access structures (all of them are based on matroids) and improve known lower bounds (see \cite{FKMP20}) for their information ratios.
	
	We use a combination of several known techniques. First of all, we apply new non-Shannon-type information inequalities (more specifically, we use the \emph{copy lemma} introduced implicitly in \cite{ZhY98};  see Section~\ref{ss:cp} for a more detailed discussion of this technique). We apply non-Shannon-type inequalities indirectly, using the linear programming approach, as proposed in \cite{FKMP20,GR19}. 
	We use very extensively the symmetry considerations. Though each of these elements was known, the combination of these methods turns out to be surprisingly efficient. The bounds we have improved are summarized in Table~\ref{table}. We believe that similar techniques can be productive not only in secret sharing but also in the other applications of non-Shannon-type information inequalities.
	
	
	
	In Section~\ref{s:inequalities} we explain the context in more detail and give the necessary formal definitions. In Section~\ref{s:Sym} we discuss symmetry tools and prove our main theorem. In the last section we present the improved bounds that we prove for secret sharing schemes.
	
	\section{Entropic Vectors and Information Inequalities}\label{s:inequalities}
	\subsection{Entropy Vectors}
Let $X=(X_i)_{1\le i\le n}$ be a sequence of jointly distributed random variables with a finite range. We denote by $h_X$ 
the \emph{entropy vector} whose coordinates  are the values of Shannon entropy for all  sub-tuples of $X$. This vector  consists of $2^n-1$ components  $h_I=H((X_i)_{i\in I})$ for each ${\varnothing\neq I\subseteq\llbracket1,n\rrbracket}$.

	Following  \cite{ZhY98}, we use the notation $\Gamma_n^*$ for the set of all entropy vectors of dimension $2^{n-1}$ (for all distributions of $n$-tuples of random variables) and  $\overline{\Gamma_n^*}$ for the closure of this set.
	
	The linear inequalities for $2^n-1$ variables that are true for all points in $\Gamma_n^*$ (interpreted as inequalities for entropies of jointly distributed random variables) are called \emph{information inequalities}. The linear combinations of those of the form 
\[
H(X_{I}) + H(X_{J}) \ge H(X_{I\cup J}) + H(X_{I\cap J})
\]	
(which is equivalent in the standard notation to the inequality $I(X_I:X_J|X_{I\cap J})\ge0$) are called Shannon-type (classical) inequalities.
	
	The vectors with $2^n-1$ coordinates (not necessarily entropic) satisfying all classical inequalities is noted $\Gamma_n$
	
	Note that $\Gamma_n^*\subset\overline{\Gamma_n^*}\subset\Gamma_n$, that $\Gamma_n^*$ is closed under addition and that $\overline{\Gamma_n^*}$ is a convex cone \cite{ZhY97}.
	\subsection{Non-Shannon-Type Inequalities}\label{ss:cp}
	There exists infirmation inequalities that are not Shannon-type. 
	The first  non-Shannon-type inequality was discovered by Zhang and Yeung \cite{ZhY98}.
	Now we know infinitely many examples of non-Shannon-type inequality. 
	All known non-Shannon-type information inequalities were proven with one of two techniques:
	the Ahlswede--K\"orner (AK) lemma (introduced in \cite{AGK76}, discussed in \cite{CK82} and \cite{K13}) or the copy lemma (implicitly used in \cite{ZhY98}, formally introduced in \cite{DFZ06}, discussed in \cite{M07a}, \cite{DFZ11} and \cite{GR19}, and with the name book ineuqualities in \cite{Cs14}).
	
	\begin{lem}[Copy Lemma]\label{l:copy}
		Let $X,Y,Z$ be three jointly distributed random vectors. There exists random vectors $X',Y', Z', Z''$ such that the following three conditions are satisfied:
		\begin{enumerate}
			\item $(X,Y,Z)$ has the same joint distribution as $(X',Y',Z')$,
			\item $(X',Z')$ has the same joint distribution as $(X',Z'')$,
			\item $Z''$ is independent of $(Y',Z')$ given $X'$, i.e.  \hbox{$I(Z'':Y',Z'|X')=0$}.
		\end{enumerate}
		We say that $Z''$ is a $Y'$-copy of $Z'$ over $X'$. By abuse of notation we identify $X',Y',Z'$ with $X,Y,Z$ respectively and say that $Z''$ is a $Y$-copy of $X$ over $Z$.
	\end{lem}
	
This lemma can be extended to the case when $X,Y,Z$ are not individual variables but tuples of jointly distributed variables.

How to use the copy lemma to prove a non-Shannon-type inequality?	
The new variables created with the copy lemma are also subject to information inequalities. Given a set of random variables, we can apply (one or many times) the copy lemma and then write down Shannon-type inequalities for all the involved random variables (the original ones along with the new ones added by copy lemma), together with the constraints on the entropy values from the conditions  of the copy lemma. For instance, if we apply the copy lemma to create a new variable $Z'$ that  is a $Y$-copy of $Z$ over $X$, we get the linear constraints
	\[ H(X,Z)=H(X,Z') \text{ and } I(Z':Y,Z|X)=0\]
We get more linear constraints if $X$ and $Z'$ are tuples that contain more than one random varibale. For example if we take a $C$-copy of $D,E$ over $A,B$, we create two new variables $D',E'$ defined by the following inequalities:
	\begin{verbatim}
	H(D)=H(D')
	H(A,D)=H(A,D')
	H(B,D)=H(B,D')
	H(A,B,D)=H(A,B,D')
	H(E)=H(E')
	H(A,E)=H(A,E')
	H(B,E)=H(B,E')
	H(A,B,E)=H(A,B,E')
	H(D,E)=H(D',E')
	H(A,D,E)=H(A,D',E')
	H(B,D,E)=H(B,D',E')
	H(A,B,D,E)=H(A,B,D',E')
	I(D',E':C,D,E|A,B)=0		
	\end{verbatim}

These linear equalities and inequalities together may imply some linear inequality for the entropy values of the initial random variables; and in some cases (if the copy lemma was applied in a clever way), the resulting information inequality can be non-Shannon-type.

	Until recently all known proofs of non-Shannon-type information inequalities could be expressed in two equivalent form: as an argument with the copy lemma and as an argument with the AK lemma. Kaced has shown \cite{K13} that this is not a pure coincidence. In fact, \emph{every} proof of a non-Shannon-type information inequality based on the AK lemma can be rephrased as an equivalent  proof of the same inequality with the copy lemma. Also, a weak conversion of this rule is true:  if a proof of an information inequality uses the copy lemma where each single instance of the ``copy'' operation creates only one new variable (technically, this means that Lemma~\ref{l:copy} is used in such a way that  $Z$ is a single variable and not a tuple of several variables), then this argument can be rephrased as an argument using the Ahlswede--K\"orner lemma.
	
	L.~Csirmaz observed in \cite{Cs21} that the copy lemma is theoretically stronger than the AK lemma. That is, some argument with the ``copying'' operation creates copies at once for $2$ (or more) random variables, then this argument in general cannot be reproduced with the technique of the AK lemma. However, we know very few evidence of this strength of the general version of the the copy lemma in natural application. To the best of our knowledge, before \cite{Cs21} the only examples of an argument that substantially uses the general version of the copy lemma were \cite{DFZ11} and \cite{GR19}.
	
	In this paper we come up with a few new applications of the general version of the copy lemma (for which an equivalent argument based on the AK lemma seems doubtful, perhaps even impossible). Moreover, we combine several important ideas, which make the set of used non-Shannon-type information inequalities more powerful:
	\begin{itemize}
		\item the general version of the copy lemma that we have discussed  above;
		\item implicit usage of new non-Shannon-type inequalities in the sense of \cite{FKMP20} and  \cite{GR19}  (the technique of linear programming permits us to apply new non-classical information inequalities without even revealing and writing them down explicitly);
		\item we extensively  use symmetries of the implied problems, which permits to decrease significantly the dimension of the relevant problem of linear programming and to reduce dramatically the computational complexity of our problem (see analysis and applications of symmetries for information inequalities in \cite{ZhT17,GR19});
		\item all our proofs are computer-assisted: we prove new results in information theory (more specifically, new lower bounds for information ratio of secret sharing schemes) that hardly can be found manually; in each case, we produce with the help of a computer, a linear program representing our problem of information theory, and then 
		obtain the required bounds using linear program solvers  (computer-assisted proofs were used in similar contexts in \cite{DFZ11,PVY13,GR19,FKMP20}).
	\end{itemize}
	None of these major ideas is new, each one was discussed and used in a similar context. However, we find surprising how efficient can be a combination of these ideas when they work together. Using the non-classical information inequalities implicitly implied by the combination of these techniques, we have improved known bounds for the information ratio  of several access structures. In what follows we discuss each of these ideas in more detail: the implicit use of non-classical inequalities in subsection~\ref{s:lower-bounds-copy-lemma}, symmetries in the section with the same name, and the different aspects of linear programming on different sections.
	\subsection{Secret Sharing and Matroids}
	
	One of the usual ``benchmarks'' for the techniques of non-classical information inequalities are lower bounds for the information ratio of secret sharing schemes.
	Let us remind that secret sharing was introduced in \cite{B79} and \cite{Sh79}.
	An \emph{access structure} for secret sharing among $n$ parties is a subdivision of all subsets of participants $\mathcal{P}(\llbracket1,n\rrbracket)$ into two classes:
	 the class of \emph{accepted} (or \emph{authorized}) and the class of \emph{not accepted} (or \emph{unauthorized}) subsets of $\llbracket1,n\rrbracket$. Each accepted subset  of parties should have an access to the secret key; no unauthorized party should have any information on the secret.  It is assumed that the property of being accepted is monotone: a superset of an accepted set is also accepted, a subset of an unauthorized set is also unauthorized.  Observe that such an access structure is determined uniquely by the family of all minimal (by inclusion) accepted subsets of parties.
	
	A standard example of an access structure is a \emph{threshold structure}: we can take as the set of accepted sets as those containing at least $\left\lfloor \frac n2\right\rfloor$ parties and accordingly define the sets that contain strictly less than $\left\lfloor \frac n2\right\rfloor$ parties as unauthorized.
	
 A \emph{secret sharing scheme} for a given access structure is defined as a joint distribution  $(s_0, s_1,\dots,s_n)$ satisfying the following conditions for each $J = \{j_1, \ldots j_k\} $:
\begin{equation}\label{eq:secret-sharing}
	\begin{array}{ll}
	H(s_0| {s_{J}})=0,& \text{ if }J \text{ is an accepted set},\\
	H(s_0| {s_{J}})=H(s_0),&  \text{ if }J \text{ is an unauthorized set}.
	\end{array}
\end{equation}
The random variable $s_0$ is 	understood as the \emph{secret key} and $s_j$ for $j=1,\ldots, n$ are the \emph{shares} given to each party.

	The \emph{information ratio} of a secret sharing scheme is the proportion of the size of the largest key to that of the secret, i.e. $\max_{i}\frac{H(s_i)}{H(s_0)}$. The information ratio of an access structure is the infimum of the information ratios of the secret sharing schemes on the access structure.

A general problem in information theory is to find the information ratio for a given access structure. It is proven by Ito, Saito and Nishizeki (see \cite{ISN89})
that any access structure as we defined above (monotone), admits a secret sharing scheme realizing it. It can easily be seen that the information ratio is at least $1$ (Let $x$ be a share, $\{x,y_1,\dots,y_i\}$ a minimal accepted set and $z$ the secret, then $H(x)\ge I(x:z|y_1,\dots,y_i)=H(z|y_1,\dots,y_i)-H(z|x,y_1,\dots,y_i)=H(z)$.). Secret sharing schemes which attain this lower bound are called \emph{ideal}. For every $n$, there exists an access structure with $n$ parties, which has an information ratio of at least $n/\log_2n$, as proven by Csirmaz \cite{Cs97}; this is the best lower bound we know for the maximal information ratio of access structures \cite{FRGR19}. 

	\subsection{Access Structures on Matroids}\label{acc_str_stu}
Among all possible access structures, one end of the spectrum is represented by the access structures that admit \emph{ideal} secret sharing schemes, i.e. those schemes where  the information ratio is equal to $1$. It is known that ideal secret sharing is closely connected with combinatorics of matroids.
(The definition of a matroid was introduced in \cite{W35}, see also \cite{O92} for a detailed introduction to the theory of matroids). 
More specifically, let $M = ( E , {\cal I} )$ be a matroid with a ground set $E$ and  the family of independent sets ${\cal I}\subset {\cal P}(E)$.
Let us fix  an element $p\in E$ and identify the other elements of the ground set with parties of a secret sharing scheme.
Following \cite{BD91}, we  define the corresponding  access structure as follows: we let  minimal accepted set be sets $C\subset E\setminus \{p\}$ such that
$C \cup \{p\}$  is a \emph{circuit} (a minimal dependent set) in the matroid.
	
It is known that ideal secret sharing is possible only if the access structure can be defined (as explained above) in terms of some matroid $M$;  on the other hand, for all access structures that cannot be defined in terms of matroids, the information ratio is at least  $\frac32$, see \cite{MFP10}. Thus, it is interesting to study the access structures in between these extremal cases: access structures defined on matroids but without ideal secret sharing. 

The matroids defined on $7$ or less points are all \emph{linearly representable}, and all corresponding access structures are known to have an ideal secret sharing. Thus, the problem of secret sharing is getting non-trivial for the access structures defined on matroids with a ground set of $8$ points. The corresponding access structure would consist of  7 participants (we will denote them  $\{1,2,3,4,5,6,7\}$, to be consistent with the notation of \cite{FKMP20}). 

We will focus on a few access structures whose study was initiated in \cite{FKMP20}. All these access structures  defined on matroids having some nice geometric interpretation. 
They are named after the matroids in \cite{O92}, from which they are derived. As usual, each access structure can be defined by their minimal authorized sets: 
\begin{table}[ht]
	\centering
	\begin{tabular}{ |c|c| }
		\hline
		\thead{Acccess Structure} & \thead{List of Minimal Authorized Sets} \\
		\hline
		$\mathcal{A}$ & \makecell{123, 145, 167, 246,\\ 257, 347, 356, 1247} \\
		\hline
		$\mathcal{A}^*$ & \makecell{123, 145, 167, 246, 257,\\ 347, 1356, 2356, 3456, 3567} \\
		\hline
		$\mathcal{F}$ & \makecell{123, 145, 167, 246, 257,\\ 347, 356, 1247, 1256} \\
		\hline
		$\mathcal{F}^*$ & \makecell{123, 145, 167, 246, 257, 1347, 1356,\\ 2347, 2356, 3456, 3457, 3467, 3567} \\
		\hline
		$\widehat{\mathcal{F}}$ & \makecell{123, 145, 167, 246, 257, 347,\\ 1256, 1356, 2356, 3456, 3567} \\
		\hline
		$\mathcal{Q}$ & \makecell{123, 145, 167, 246, 257, 347,\\ 1247, 1256, 1356, 2356, 3456, 3567} \\
		\hline
		$\mathcal{Q}^*$ & \makecell{123, 145, 167, 246, 257, 1247, 1347,\\ 1356, 2347, 2356, 3456, 3457, 3467, 3567} \\
		\hline
	\end{tabular}
\end{table}

The ultimate goal is to find the optimal information ratio for each of these structures (and study the connection of optimal ratio with the combinatorial properties of matroids). This goal was not achieved in 	\cite{FKMP20}, and it is not achieved in our work. However, we take a new step in this direction and improve the known lower bound for optimal information ratio 
of these $7$ access structures.

	\subsection{Lower Bounds for Secret Sharing  From the Copy Lemma}
	\label{s:lower-bounds-copy-lemma}
	
How to prove a lower bound for the information ratio of a certain access structure? The usual approach uses  the technique of information inequalities. We write down the equalities \eqref{eq:secret-sharing} that define the access structure, and all Shannon type inequalities for the involved random variables, and then try to combine these equalities and inequalities to derive a conclusion 
\begin{equation} \label{eq:info-ratio}
 \max \{ H(s_i) \} \ge r\cdot H(S_0)
\end{equation}
for a certain real number $r$. If we succeed, this means that the information ratio of this access structure is at least $r$. Such an argument can be found, for example, in \cite{CSGV93}.

For some access structures the proposed simple scheme can be improved: we can add non-Shannon-type inequalities for entropies of the involved random variables; the new (non-classical) constraint may help to prove \eqref{eq:info-ratio} with a larger value of $r$. Proofs following this scheme can be found, e.g., in \cite{BLP08} and \cite{MB11}.

As we prove new non-Shannon-type inequalities with the help of the copy lemma, the entire proof of a lower bound for the information ratio can be subdivided into two big steps:

\smallskip
\noindent
\emph{Step 1:} Apply one or several times the copy lemma, write down Shannon-type inequalities for the involved random variables, and deduce a new non-Shannon-type inequality.

\smallskip
\noindent
\emph{Step 2:} Write the conditions \eqref{eq:secret-sharing}, write down classical inequalities for Shannon entropies of the involved random variables, add a non-Shannon-type inequality proven on Step~1, and deduce \eqref{eq:info-ratio} for some specific $r$.
\smallskip

Historically,  Step~1 and Step~2 were often done in two different publications (typically done by different groups of authors), often with years between the first and the second one
(see the section 6 of \cite{MFP10}). The main disadvantage of this approach is that we cannot know in advance which new inequality should be proven on Step~1 so that it could be useful to eventually use in Step~2. It turns out that the two steps can be merged together, so that we do not even need to reveal explicitly the ``useful'' non-Shannon-type inequality:

\smallskip
\noindent
\emph{Merged steps 1+2:} Write the conditions \eqref{eq:secret-sharing}, apply one or several times the copy lemma, write down Shannon-type inequalities for the involved random variables, and deduce $\eqref{eq:info-ratio}$ for some specific $r$.
\smallskip

This type of argument was discussed in detail in \cite{GR19}. A similar approach (with the AK lemma instead of the copy lemma) was used earlier in \cite{FKMP20, BBEFP20}.

The mentioned approach has an important inherent disadvantage: 	each new random variable added by the copy lemma doubles the dimension of the final linear program, which increases dramatically the computational complexity of the problem. There is one tool that helps to mitigate this flaw: reduce the required applications of the copy lemma and decrease the dimension of the problem of linear programming. This tool is symmetries.

The symmetries of the random variables in \eqref{eq:secret-sharing} can be preserved despite the use of non symmetric applications of copy lemma, by adding symmetry conditions to the linear program. These extra conditions may improve the $r$ in \eqref{eq:info-ratio} if the copy lemma applications in consideration are not symmetric. More detailed discussion of this method is given in the next section, precisely from Theorem~\ref{thm:sym} to (including) subsection~\ref{ss:a1}.
	
	\section{Symmetries}\label{s:Sym}
		\subsection{Entropic Vectors and the Existence of Symmetric Solutions}
			In this subsection we talk about the use of symmetries on the sets of almost entropic vectors. 
			In what follows we assume that we are focused on the properties of some ``symmetric'' subset $E$ of $\overline{\Gamma_n^*}$:
			\begin{definition}
				We say that $E$ is \emph{symmetric for a permutation} $\sigma$ if for every element $h=(h_{\{1\}},\dots,h_{\{1,\dots,n\}})\in E$, we have $\sigma\cdot h=(h_{\sigma\cdot\{1\}},\dots,h_{\sigma\cdot\{1,\dots,n\}})\in E$. We say that it is \emph{symmetric for a permutation group} $G<Sym_n$ if it is symmetric for every element of $G$.
				
				Similarly, we say that an element $f$ of the dual space (a linear form on $E$) is invarinat for $\sigma$ if $f\cdot h=f\cdot(\sigma^{-1}\cdot h)$ for all $h$. 
				
			\end{definition}
			
			\begin{lem}\label{vec} Let $E$ be convex and symmetric for a group $G$ and suppose we want to optimize a scalar product $f\cdot h$ with $h\in E$ and that $f$ is invariable under $G$. Then we can restrain this problem on a subset of $E$ which may have lower dimension: $\max f \cdot h$ over $E$ is equal to the maximum of this linear form on the subset of $E$ that respects the symmetries, namely
 \[E\cap\{h\mid \forall\sigma\in G, \sigma\cdot h=h\}\]
 
			\end{lem}
		\begin{proof}
			If a point $h'\in E$ is an optimal solution and $f\cdot h'=a$, then the same optimal value is achieved on every element in $h'^G$ (the orbit of $h'$ under actions by elements of $G$), because $f$ is invariant under $G$. All this orbit belongs to $E$ since $E$ is symmetric for $G$. By convexity of $E$, the symmetric vector $h = \frac1{|G|}\sum_{\sigma\in G}\sigma\cdot h'$ belongs to $E$, and $f\cdot h = \frac1{|G|}f\cdot\sum_{\sigma\in G}\sigma\cdot h' = a$. Thus, the optimal value is attained on a symmetric vector.
			\end{proof}
		\begin{rem}\label{last_par_vec}
			Note that even if $E$ is not symmetric for $G$, if there exists a symmetric (for $G$) subset of $E$ that contains an optimal (for $f$) vector, then we can still use these symmetries ($G$) on $E$ without needing to explicitly express that subset.
		\end{rem}

		\subsection{Symmetries for Linear Programs}
		
		In this subsection, we combine the symmetry considerations with technique from Section~\ref{s:lower-bounds-copy-lemma}. 
		
	
	The linear programs for finding lower bounds on secret sharing were discussed in \cite{PVY13} and used with copy lemmas in \cite{GR19}, here we reformulate it as a proposition:
		\begin{prop}\label{p:linear-program-for-ratio}
			Let $A$ be an access structure with $r<n$ participants and $(s_0, s_1,\ldots, s_r)$ be a secret sharing scheme for this access structure. We extend this distribution by adding  random variables $s_{r+1}, \ldots, s_{n-1}$.
			 The linear program $P$ described below provides a lower bound on the information ratio of $A$:
			
				$\min  x$
				subject to
				\begin{enumerate}[label=\roman*]
				\item $x\ge h_T$ for every singleton $T\subseteq\{2,\dots,r+1\}$
				\item classical information inequalities for $(h_S)_{\varnothing\neq S\subseteq\{1,2,\dots,n\}}$
				\item the information equalities \eqref{eq:secret-sharing} for the entropies of $(s_0, s_1,\ldots, s_r)$ that define 
				 the access structure $A$,
				 \item the normalization condition $h_{\{1\}}=1$.
				\item the equalities for entropies that define each of the random  variables  $s_{r+1}, \ldots, s_{n-1}$ as a copy of  other variables (with smaller indices), obtained by an instance of an application of the copy lemma.
				\end{enumerate}
		\end{prop}
		
		We will refer to the sets of conditions above as items (i), (ii), (iii), (iv) and (v) in the rest of this subsection. 
	
		We will apply the following  proposition to 7 access structures mentioned in the subsection~\ref{acc_str_stu}.

	
Let us explain now how we can use symmetries to ``amplify'' the linear programs providing lower bounds for information ratio of secret sharing schemes.
\begin{thm}\label{thm:sym}
Let $P$ be a linear program from Proposition~\ref{p:linear-program-for-ratio} for  some access structure, and let
\hbox{$G<Sym(\{2,3,\dots,r+1\})$} 
be the symmetry group of the access structure mentioned in the item (iii) at the beginning of this subsection. Suppose the objective function is symmetric under $G$. Then we can add the constraints  $h_S=h_{\sigma\cdot S}$ for all $\sigma\in G$ and $S\subseteq\{1,2,\dots,r+1\}$ as extra conditions to the linear program $P$. The resulting linear program provides a lower bound on the minimal ratio of this access structure.
\end{thm}
\begin{proof}

Let $E$ be the closure of the set of almost entropic vectors in $\overline{\Gamma_{r+1}^*}$ that correspond to the secret sharing schemes of the access structure; i.e. vectors $(h_I)_{\varnothing\neq I\subseteq\{1,\dots r+1\}}$ satisfying the linear constraints (iii), (when \hbox{$\varnothing\neq I\subseteq\{2,\dots,r+1\}$} is an accepted set, $h_{\{1\}\cup I}-h_I=0$ and otherwise \hbox{$h_{\{1\}\cup I}-h_I=h_1$).}  This set is convex because it's a subset of the convex cone, defined by linear inequalities (if two vectors satisfy the same linear inequalities, so does their weighted average), moreover it's symmetric for $G$ as $\overline{\Gamma_{r+1}^*}$ is symmetric for $Sym_{r+1}\ge G$ and that the inequalities (iii) are symmetric for $G$ by definition.

The minimal information ratio can be (purely theoretically) described  as the optimal value of the following linear program: We take $2^{r+1}-1$ variables  $H_I$, $\varnothing\neq I\subseteq\{1,\dots,r+1\}$ for the components of entropic vectors in $\overline{\Gamma_{r+1}^*}$ and assume these vectors belong to $E$ (which is true for all secret sharing schemes on the access structure). Then we add one more variable $x$ to the linear program, with (i) as constraints; at last, we add the normalization condition $H_1=1$ (iv). Now, the minimal value of $x$ gives the optimal information ratio of the scheme.

Since the set $E$ is symmetric for the group $G$, we can add to our linear program the symmetry constraints $H_S=H_{\sigma\cdot S}$ for all $\sigma\in G$ and $S\subseteq \{1,\dots r+1\}$ without changing the optimal value of the objective function.

So far this argument was purely theoretical: we do not have a complete characterization of $\overline{\Gamma_{r+1}^*}$, and therefore we do not have a complete description of $E$.
(Even if such a characterization is eventually found, it will contain infinitely many linear constraints, see \cite{M07i}). In a more realistic setting, we take a linear program that contains not a precise description of $\overline{\Gamma_{r+1}^*}$ but only that of a convex superset, with some finite family of information inequalities, for example all the classical information inequalities. We can keep the constraints (i), (iii) and (iv); together with a subset of the constraints for $\overline{\Gamma_{r+1}^*}$, these will define a set $E' \supset E$. The minimal value of the objective function on  $E' $ can be smaller than on $E$, so our linear program may give not necessarily the exact value of the information ratio but only a lower bound for it. Note that we can also take all symmetry constraints that are valid for $E$ as stated in Remark~\ref{last_par_vec} too.

Observe that we can extend a distribution $(X_1, \ldots , X_{r+1})$ that represents a secret sharing scheme by adding a few new random variables $X_{r+2}, \ldots, X_n$ applying several times the copy lemma. For the entropies of the extended distribution we can write the constraints in the item (v) (the conditions for the new variables from the copy lemma) and  the information inequalities that are valid for the extended profile $(X_I)$ with $I\subseteq \{1,\ldots n\}$. In particular, we may take all classical information inequalities for all involved random variables.

Observe that we can keep the symmetry constraints only for the coordinates of the entropic vectors $(X_I)$ with $I\subseteq \{1,\ldots r+1\}$ and not for the coordinates $X_{r+2},\dots X_n$ introduced with the copy lemma, as the latter may be defined based on the former in a way to break the symmetries. For instance if $(23)\in G$, and we take a copy $X'_2$ of $X_2$, its relation to $X_3$ will not be the same as to $X_2$.

The linear program described this way, which provides a lower bound to the minimal information ratio of the access structure, indeed corresponds to the one stated in the theorem: We initially had the items (i), (ii), (iv) and a part of the item (ii), to which we have legitimately added the symmetries, and to the resulting linear program we have added the item (v) as well as the extension to full item (ii).

%
\end{proof}	

\begin{rem}\label{r:sym_help}
	If the linear program is symmetric for a group $G$, then the symmetry constraints will not change the optimal value. However, even in this case the symmetry constraints can be useful:  they may reduce the dimension of the linear program, therefore improve the computational complexity of the problem.
\end{rem}

\subsection{Adding symmetric version of the copy lemma}\label{ss:a1}

Another way of proving the Theorem~\ref{thm:sym} is based on the observation that we could make a linear program that is perfectly symmetric, including the constraints from the item (v). To this end we would need to add all ``symmetric'' forms of the used instances of the copy lemma. In practice we do not do this, since this would increase dramatically the  dimension of the linear program. However this argument explains and justifies Remark~\ref{r:sym_help} directly.

In effect the items (i) to (iv) are symmetric under the symmetry group $G$. Only the item (v) is not necessarily symmetric. However we can add all the symmetric copy lemmas to make the item (v) symmetric. Let us call this symmetric version of the linear program $P'$ (note that $P'$ may have a larger $n$ thus might be too costly in time to solve). As $P'$ is stable under $G$ (i.e. $G$ is a subgroup of the symmetry group of the set of solution vectors defined by $P'$) we can add the conditions claimed in this theorem to $P'$ and we still get a lower bound for the intended quantity. Adding these symmetry conditions to $P$ is legitimate from Remark~\ref{last_par_vec}. Another way to see  this is to simply add the symmetries to $P'$ and then delete the symmetric copy lemmas from the resulting linear program to obtain what we want. Deleting conditions will only decrease the objective value as we minimize the objective function. Therefore it is legitimate to add these conditions to $P$ and that they cannot give a better bound than the objective value of $P'$.

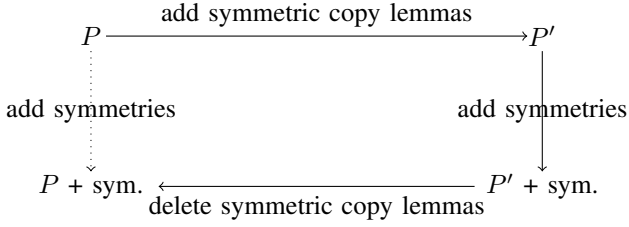
\smallskip

\begin{tikzpicture}
\draw (0,2) node {$P$};
\draw (6,2) node {$P'$};
\draw (0,0) node {$P$ + sym.};
\draw (6,0) node {$P'$ + sym.};
\draw[->] (0.2,2) -- (5.8,2) node[pos=0.5, align=center, above]{add symmetric copy lemmas};
\draw[->] (6,1.8) -- (6,0.2) node[pos=0.5, align=right] {add symmetries};
\draw[->] (5.1,0) -- (0.9,0) node[pos=0.5, align=center, below]{delete symmetric copy lemmas};
\draw[dotted,->] (0,1.8) -> (0,0.2) node[pos=0.5, align=left] {add symmetries};
\end{tikzpicture}
The objective values of these linear programs compare as: $val(P)\le val(P+sym)\le val(P'+sym)=val(P')$

		\subsection{Symmetry Group of Access Structures}\label{ss:Groups}
	The symmetry groups of these access structures are:
	\begin{itemize}
		\item $\mathcal{A}, \mathcal{A}^*$: $\langle(12)(56),(14)(36),(17)(35)\rangle$
		\item $\mathcal{F}, \mathcal{F}^*$: $\langle(12)(4576),(46)(57)\rangle$
		\item $\mathcal{Q}, \widehat{\mathcal{F}}, \mathcal{Q}^*$: $\langle(12)(47),(12)(56)\rangle$
	\end{itemize}

It is easy to check that these access structures are invariant under these permutations. In the following subsections we explain how to find for each access structure its group of symmetries.

\subsubsection{For $\mathcal{A}$ and $\mathcal{A}^*$ }

For $\mathcal{A}$, we can present the structure as in the image below:

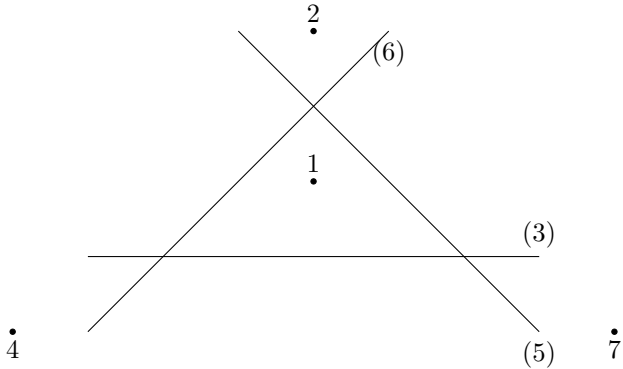
\begin{tikzpicture}
\filldraw (0,0) circle (1pt) node[align=center, above] {$1$};
\filldraw (0,2) circle (1pt) node[align=center, above] {$2$};
\draw (-3,-1) -- (3,-1) node[align=right, above]{$(3)$};
\filldraw (-4,-2) circle (1pt) node[align=left, below]{$4$};
\draw (-1,2) -- (3,-2) node[align=right,below]{$(5)$};
\draw (-3,-2) -- (1,2) node[align=right,below]{$(6)$};
\filldraw (4,-2) circle (1pt) node[align=right,below]{$7$};
\end{tikzpicture}

\smallskip

Any two points with the line according to which they are at the same side is a minimum accepting set as well as all four points. In fact points are exactly those that appear more times in the list of minimal authorized sets and lines are those that appear less times. Now for any permutation on $\{1,2,4,7\}$, we have a permutation on $\{3,5,6\}$ such that the minimum authorized sets (described geometrically) are preserved. Let us take a generating set of $Sym(\{1,2,4,7\})=\langle(12),(14),(17)\rangle$, we have transpositions $(56)$, $(36)$ and $(35)$ respectively, which make them preserve the minimum authorized sets. Thus $\langle(12)(56),(14)(36),(17)(35)\rangle$ generate the symmetry group of $\mathcal{A}$.
The same argument works for $\mathcal{A}^*$, we just change the geometric definition of minimum authorized sets to ``any two points with the line according to which they are at the same side as well as any point with all three lines''.

\subsubsection{For $\mathcal{F}$ and $\mathcal{F}^*$}

For $\mathcal{F}$ we use the following geometric presentation:

\smallskip

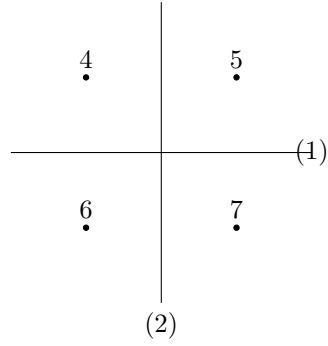
\begin{tikzpicture}
\draw (-2,0) -- (2,0) node[align=right] {$(1)$};
\draw (0,2) -- (0,-2) node[align=center, below] {$(2)$};
\filldraw (-1,1) circle (1pt) node[align=left,above] {$4$};
\filldraw (1,1) circle (1pt) node[align=left,above] {$5$};
\filldraw (-1,-1) circle (1pt) node[align=left,above] {$6$};
\filldraw (1,-1) circle (1pt) node[align=left,above] {$7$};
\end{tikzpicture}

\smallskip

In fact $3$ is the only one not appearing in the minimum authorized sets of size four, thus it must be stable. $1$ and $2$ are the only ones appearing in both of them. Hence we have three orbits. Looking at the minimum authorized sets of size three, we get the rest of the image.  The authorized sets are:
\begin{itemize}
	\item with $3$:
	\begin{itemize}
		\item both lines
		\item two points separated by both lines
	\end{itemize}
	\item without $3$:
	\begin{itemize}
		\item two points with a line not separating them
		\item both lines with two points separated by both lines
	\end{itemize}
\end{itemize}
The symmetries are mirror images and rotations for this image, therefore they are generated by $(12)(4576)$ and $(46)(57)$. The same argument works for $\mathcal{F}^*$, we just change the interpretation of the image to define the minimum authorized sets as
\begin{itemize}
	\item ``with $3$:
	\begin{itemize}
		\item both lines
		\item a line with two points separated by both lines
		\item any three points
	\end{itemize}
	\item without $3$: two points with a line not separating them''.
\end{itemize}

\subsubsection{For $\mathcal{Q}$, $\widehat{\mathcal{F}}$ and $\mathcal{Q}^*$}

For $\mathcal{Q}$, $\widehat{\mathcal{F}}$ and $\mathcal{Q}^*$ we have four orbits: $\{1,2\}$, $\{3\}$, $\{4,7\}$ and $\{5,6\}$.

For $\mathcal{Q}$, the number of times the keys appear in minimum authorized sets of size four reveal the restrictions on the orbits: $1$ and $2$ three times, $3$ four times, $4$ and $7$ twice, $5$ and $6$ five times.

For $\widehat{\mathcal{F}}$, $3$ is the only one appearing four times among minimum authorized sets of size four, $1$ and $2$ appear twice, $5$ and $6$ five times, and $4$ and $7$ once.

For $\mathcal{Q}^*$,  $3$ is stable because it is the only one to appear only once in minimum accepted sets of size three. $1$ and $2$ are the only ones to be with $3$ in this appearance. $4$ and $7$ are the only ones to appear both with $1$ and $2$ in a single minimum authorized set of size four.

Hence the symmetry group is $\langle(12)(56),(12)(47)\rangle$. In fact, it can be checked that both generating elements preserve minimum authorized sets.

\subsection{Inherent sanity checks}

\begin{table*}[t]
	\centering
	\begin{tabular}{ ||c||c||c|c|c|| }
		\hline
		\thead{Access\\structure}&\thead{known lower bound\\based on AK lemma \cite{FKMP20}}&\thead{how we use copy lemma}&\thead{bounds we prove\\using symmetries}&\thead{weaker bounds we can\\prove without symmetries}\\
		\hline
		$\mathcal{A}$ & $\sfrac98=1.125$ & $0,3,4,7|1,2,5,6$ & $\sfrac{57}{50}=1.14$ & $\sfrac{135}{119}=1.134\dots$ \\
		\hline 
		$\mathcal{A}^*$ & $\sfrac{33}{29}=1.137\dots$ & \makecell{$5,6$-copy$(0,3|1,2,4,7)$ and\\ $0,0',3,3'$-copy$(1,2|4,5,6,7)$} & $\sfrac{52}{45}=1.1\overline{5}$ & $\sfrac{33}{29}=1.137\dots$ \\
		\hline
		$\mathcal{F}$ & $\sfrac98=1.125$ & $0,2,4,6|1,3,5,7$ & $\sfrac{17}{15}=1.1\overline{3}$ & $\sfrac{26}{23}=1.130\dots$\\
		\hline
		$\mathcal{F}^*$ & $\sfrac{42}{37}=1.\overline{135}$ & \makecell{$3,7$-copy$(0,4|1,2,5,6)$ and\\ $0,0',4',5$-copy$(1,4|2,3,6,7)$} & $\sfrac87=1.142\dots$ &  $\sfrac{42}{37}=1.\overline{135}$ \\
		\hline
		$\widehat{\mathcal{F}}$ & $\sfrac{42}{37}=1.\overline{135}$ & \makecell{$2,6$-copy$(0,4|1,3,5,7)$ and\\ $0,0',4',5$-copy$(1,4|2,3,6,7)$} & $\sfrac{23}{20}=1.15$ & $\sfrac{42}{37}=1.\overline{135}$\\ 
		\hline
		$\mathcal{Q}$ & $\sfrac98=1.125$ & \makecell{$0,2,4,6$-copy$(t,v|1,3,5,7)$ and\\$0,2,4,6,t',v'$-copy$(t,v|1,3,5,7)$\\with $t=(0,4)$ and $v=(2,6)$} & $\sfrac{17}{15}=1.1\overline{3}$ & $\sfrac{17}{15}=1.1\overline{3}$ \\ 
		\hline
		$\mathcal{Q}^*$ & $\sfrac{33}{29}=1.137\dots$ & \makecell{$3,7$-copy$(0,4|1,2,5,6)$and\\$0,0',4,4'$-copy$(1,5|2,3,6,7)$} & $\sfrac87=1.142\dots$ & $\sfrac{33}{29}=1.137\dots$ \\
		\hline
	\end{tabular}
	\smallskip
	\caption{Access structures for which we have improved lower bounds on the information ratio.}\label{table}
\end{table*}

\begin{rem}
	In practice, we can make errors while programming the symmetries for an access structure. However, there is fortunate ``sanity check'' embedded in this method.
	If there is an error in the symmetry group added as conditions to the linear program (if we use a wrong  group $H$ instead of the symmetry group $G$ of our access structure),
	then two cases are possible:
	\begin{itemize}
		\item either $H<G$, then we still get a valid lower bound, possibly worse than what we could achieve with the true group of symmetries of this access structure;
		\item or $H\nless G$, then we get an unfeasible program (with no solution). Indeed,  any element of $H\setminus G$ applied to the equalities in the item (iii) gives a contradiction. Namely, if $A$ is an accepted set but $\sigma\cdot A$ for a $\sigma\in H\nless G$ isn't, then as $H(secret|A)=0$, $H(secret|\sigma\cdot A)=0$ too, but this contradicts $H(secret|\sigma\cdot A)=H(secret)$. 
	\end{itemize}
\end{rem}		

	\section{Main results}

	Using the symmetries (Theorem~\ref{thm:sym}) and the copy lemma (Lemma~\ref{l:copy}), we have improved several lower bounds on the information ratios of the access structures (on matroids) given in \cite{FKMP20}.
	
	In the following theorem $\sigma(\Gamma)$ denotes the information ratio of the access structure $\Gamma$.
	\begin{thm}[Main Result]
		For the seven access structures defined in Section~\ref{acc_str_stu} we have the following lower bounds for their information ratios:
		\begin{itemize}
			\item $\sigma(\mathcal{A})\ge\sfrac{57}{50}=1.14$
			\item $\sigma(\mathcal{A}^*)\ge\sfrac{52}{45}=1.1\overline{5}$
			\item $\sigma(\Gamma)\ge\sfrac{17}{15}=1.1\overline{3}$ for $\Gamma\in\{\mathcal{F},\mathcal{Q}\}$
			\item $\sigma(\Gamma)\ge\sfrac87=1.142\dots$ for $\Gamma\in\{\mathcal{F}^*,\mathcal{Q}^*\}$
			\item $\sigma(\widehat{\mathcal{F}})\ge\sfrac{23}{20}=1.15$
		\end{itemize}
		See Table~\ref{table}
	\end{thm}
	\begin{proof}
		We prove these bounds with the help of computer. In what follows we explain the scheme of our proofs so that they can be reproduced independently, with any linear program solver.
	
	For each of the seven access structures we construct a linear program as explained in Proposition~\ref{p:linear-program-for-ratio}. In this linear program we use auxiliary random variables with one or two applications of the copy lemma (see column 3 of the Table~\ref{table}), and we add the constraints to express for each access structure the symmetry conditions (see Section~\ref{ss:Groups}). In this table $0$ represents the secret, and $1,2,3,4,5,6,7$ are the shares of the parties;  $Z|X$ denotes ``copy of $Z$ over $X$'' and $Y$-copy$(Z|X)$ denotes ``$Y$-copy of $Z$ over $X$''. In the column 4 of the Table~\ref{table} we show the resulting lower bound for the information ratio of each access structure. For comparison, in column 5 we show weaker bound (strictly weaker except for $\mathcal{Q}$) that can be proven with the same usage of the copy lemma but without symmetry conditions.
	
	Since the needed linear programs (described in Proposition~\ref{p:linear-program-for-ratio}) are too large to type them by hand, we generate them by a computer. To find the optimal values of these programs we use a linear program solver software \emph{Gurobi Optimizer 9.1.2} \cite{Gurobi}, and check the obtained bounds with an exact (rational) linear programs solver \emph{QSopt\_ex} \cite{QSopt_ex} based on the \emph{GNU Multiple Precision Arithmetic Library} \cite{GMP} and \emph{GNU Linear Programming Kit} \cite{GLPK}.
	
\end{proof}

	\section*{Acknowledgment}
	The author thanks Andrei Romashchenko for useful discussions.
\balance


\begin{thebibliography}{99}
		\bibitem{AGK76} \emph{Bounds on Conditional Probabilities with Applications in Multi-User Communication}, Rudolf Ahlswede, Peter G\'acs, J\'anos K\"orner, Probability Theory and Related Fields, volume 34, issue 2, 1976.
		\bibitem{B79} \emph{Safeguarding Cryptographic Keys}, G. Robert Blakley, Managing Requirements Knowledge, 1979.
		\bibitem{BBEFP20} \emph{Common Information, Matroid Representation and Secret Sharing for Matroid Ports}, Michael Bamiloshin, Aner Ben-Efraim, Oriol Farr\`as, Carles Padr\'o, Designs, Codes and Cryptography, 2020.
		\bibitem{BD91} \emph{On the Classification of Ideal Secret Sharing Schemes}, Ernest F. Brickell and Daniel M. Davenport, Journal of Cryptology, Volume 4, Issue 2, pages 123-134, 1991.
		\bibitem{BLP08} \emph{Matroids Can Be Far from Ideal Secret Sharing}, Amos Beimel, Noam Livne, Carles Padr\'o, Theory of Cryptography Conference, 2008.
		\bibitem{CK82} \emph{Information Theory: Coding Theorems for Discrete Memoryless Systems}, Imre Csisz\'ar, J\'anos K\"orner, Academic Press, 1982.
		\bibitem{Cs97} \emph{The Size of a Share Must Be Large}, L\'aszl\'o Csirmaz, Journal of Cryptology, volume 10, issue 4, 1997.
		\bibitem{Cs14} \emph{Book Inequalities}, L\'aszl\'o Csirmaz, IEEE Transactions on Information Theory, volume 60, issue 11, 2014.
		\bibitem{Cs21} \emph{Exploring the Entropic Region}, L\'aszl\'o Csirmaz, talk at Institute of Information Theory and Automation, Prague, October 2021.
		\bibitem{CSGV93} \emph{On the Size of Shares for Secret Sharing Schemes}, Renato M. Capocelli, Alfredo De Santis, Luisa Gargano, Ugo Vaccaro, Journal of Cryptology, volume 6, issue 3, 1997.  
		\bibitem{DFZ06} \emph{Six New Non-Shannon Information Inequalities}, Randall Dougherty, Christopher Freiling, Kenneth Zeger, ISIT 2006.
		\bibitem{DFZ11} \emph{Non-Shannon Information Inequalities in Four Random Variables}, Randall Dougherty, Christopher Freiling, Kenneth Zeger, arXiv, \url{https://arxiv.org/pdf/1104.3602.pdf}, 2011.
		\bibitem{FKMP20} \emph{Improving the Linear Programming Technique in the Search for Lower Bounds in Secret Sharing}, Oriol Farr\`as, Tarik Kaced, Sabasti\`a Mart\'in, Carles Padr\'o, IEEE Transactions on Information Theory, Vol. 66, No. 11, November 2020.
		\bibitem{FRGR19} \emph{Local Bounds for the Optimal Information Ratio of Secret Sharing Schemes}, Oriol Farr\`as, Jordi Ribes-Gonz\'alez, Sara Rici, Designs, Codes and Cryptography, volume 87, 2019.
		\bibitem{GLPK} \emph{GNU Linear Programming Kit}, \url{https://www.gnu.org/software/glpk/}
		\bibitem{GMP} \emph{GNU Multiple Precision Arithmetic Library}, \url{https://gmplib.org/}
		\bibitem{GR19} \emph{How to Use Undiscovered Information Inequalities: Direct Applications of Copy Lemma}, Emirhan G\"urp\i nar, Andrei Romashchenko, IEEE ISIT 2019.
		\bibitem{Gurobi} \emph{Gurobi Optimizer}, \url{https://www.gurobi.com}.
		\bibitem{ISN89} \emph{Secret Sharing Scheme Realizing General Access Structure}, Mitsuru Ito, Akira Saito, Takao Nishizeki, Electronics and Communications in Japan, part 3, volume 72, no. 9, 1989 (Japanese publication 1988).
		\bibitem{K13} \emph{Equivalence of Two Proof Techniques for Non-Shannon-type Inequalities}, Tarik Kaced, ISIT 2013.
		\bibitem{M07a} \emph{Adhesivity of Polymatroids}, Franti\u sek Mat\'u\u s, Discrete Mathematics, volume 307, issue 21, 2007.
		\bibitem{M07i} \emph{Infinitely Many Information Inequalities}, Franti\u sek Mat\'u\u s, ISIT 2007.
		\bibitem{MB11} \emph{Improved Upper Bounds for the Information Rates of the Secret Sharing Schemes Induced by the V\'amos Matroid}, Jessica Ruth Metcalf-Burton, Discrete Mathematics, Volume 311, 2011.
		\bibitem{MFP10} \emph{On Secret Sharing Schemes, Matroids and Polymatroids} Jaume Mart\'i-Farr\'e, Carles Padr\'o, Journal of Mathematical Cryptology, volume 4, 2010.
		\bibitem{O92} \emph{Matroid Theory}, James G. Oxley, Oxford University Press, 1992.
		\bibitem{QSopt_ex} \emph{QSopt\_ex}, David Applegate, William Cook, Sanjeeb Dash and Daniel Espinoza, \url{https://www.math.uwaterloo.ca/~bico/qsopt/ex/}.
		\bibitem{P13} \emph{Lecture Notes in Secret Sharing}, Carles Padr\'o, \url{https://web.mat.upc.edu/carles.padro/arc02v03.pdf}, 2013.
		\bibitem{PVY13} \emph{Finding Lower Bounds on the Complexity of Secret Sharing Schemes by Linear Programming}, Carles Padr\'o, Leonor V\'azquez, An Yang, Discrete Applied Mathematics, volume 161, issue 7-8, 2013.
		\bibitem{Sh79} \emph{How to Share a Secret}, Adi Shamir, Communications of the ACM, volume 22, number 11, 1979.
		\bibitem{W35} \emph{On the Abstract Properties of Linear Dependence}, Hassler Whitney, American Journal of Mathematics, volume 57, no.3, 1935.
		\bibitem{ZhT17} \emph{On the Symmetry Reduction of Information Inequalities}, Kai Zhang, Chao Tian, IEEE Transactions on Communications, 2017.
		\bibitem{ZhY97} \emph{A Non-Shannon-Type Conditional Inequality of Information Quantities}, Zhen Zhang, Raymond W. Yeung, IEEE Transactions on Information Theory, 1997.
		\bibitem{ZhY98} \emph{On Characterization of Entropy Function via Information Inequalities}, Zhen Zhang, Raymond W. Yeung, IEEE Transactions on Information Theory, 1998.
	\end{thebibliography}
\end{document}